\documentclass[11pt]{article}
\pdfoutput=1
\usepackage{hyperref}

\usepackage[utf8]{inputenc}

\usepackage{geometry}
\usepackage{graphicx} 
\graphicspath{ {} }

\def\Adj{{\mathop{\rm Adj}}}
\def\Det{{\mathop{\rm Det}}}
\def\M#1#2{#1^{#2\times #2}}

\newcommand{\bigO}{\mathcal{O}}
\def\A#1#2{A_{#1,#2}}

\def\sumtoJ{\sum_{i=1}^{J} m_i}
\def\sumtoJplusone{\sum_{i=1}^{J+1} m_i}

\usepackage{amsthm}
\newtheorem{thm}{Theorem}

\newtheorem{remark}[thm]{Remark}
\newtheorem{defn}[thm]{Definition}
\newtheorem{ntn}[thm]{Notation}

\usepackage{ upgreek }

\usepackage[ruled]{algorithm}
\usepackage{algorithmic}

\usepackage{ dsfont }

\usepackage{booktabs} 
\usepackage{array} 
\usepackage{paralist} 
\usepackage{verbatim} 
\usepackage{subfig} 

\usepackage{fancyhdr} 
\usepackage{sectsty}
\allsectionsfont{\sffamily\mdseries\upshape} 

\usepackage[nottoc,notlof,notlot]{tocbibind} 
\usepackage[titles,subfigure]{tocloft} 

\usepackage[backend=biber, style=alphabetic]{biblatex}  
\title{On Fast Matrix Inversion via Fast Matrix Multiplication}
\author{Zak Tonks \\ Department of Computer Science, University of Bath \\ z.p.tonks@bath.ac.uk }

\begin{document}
\maketitle

\begin{abstract}

Volker Strassen first suggested an algorithm \cite{Strassen1969} to multiply matrices with worst case running time less than the conventional $\bigO(n^3)$ operations in 1969. He also presented a recursive algorithm with which to invert matrices, and calculate determinants using matrix multiplication. James R. Bunch \& John E. Hopcroft improved upon this in 1974 \cite{BunchHopcroft1974} by providing modifications to the inversion algorithm in the case where principal submatrices were singular, amongst other improvements. We cover the case of multivariate polynomial matrix inversion, where it is noted that conventional methods that assume a field will experience major setbacks. Initially, there existed a presentation of a fraction free formulation of inversion via matrix multiplication along with motivations in \cite{TonksSankaranDavenport2017}, however analysis of this presentation was rudimentary. We hence provide a discussion of the true complexities of this fraction free method arising from matrix multiplication, and arrive at its limitations.

\end{abstract}

\section{Background}

Throughout, we discuss matrices over a ring of polynomials, say $R[x_{1},...,x_{m}]$, where $m>1$. Until stated otherwise, $A$ is a matrix of polynomials, and we assume $A$ to be square, of dimension $n = 2^m$ for some $m \in \mathds{N}$. We will refer to this as being of ``\emph{binary size}'', and respectively matrices of ``\emph{non binary size}'' are those that do not fit this description. In particular matrices of binary size can be divided up into four equal quadrants, and as such $A_{i,j}$, $i,j \in \{1,2\}$ refers to the canonical quadrant unless stated otherwise. When we refer to the size of a matrix, we mean either of its dimensions (as every matrix traversed is square unless stated otherwise). Hence if we say ``a matrix of half the size'' we mean a matrix of size $\frac{n}{2} \times \frac{n}{2}$. If we say $A$ is of degree $d$, we mean that all its entries are polynomials of degree $d$, and when we refer to to the degree of a polynomial, we refer to its total degree in all variables, rather than any variable in particular. We assume the result of any multiplication or addition of two polynomials to be expanded, ie. we distribute multiplication over addition. We do not assume polynomial arithmetic to be unit cost, and note that the cost of multiplication of two multivariate polynomials to addition of the same pair should be high, and this ratio even increases with degree. For the time being, we assume every entry of any one matrix to have exactly degree $d$, ie. ``uniformly of degree $d$''.

\subsection{Fast Matrix Multiplication and Inversion\label{sect:FFMultInversion}}

Strassen-like matrix multiplication has been of interest since Strassen's first presentation in 1969. In particular these algorithms are able to provide methods to multiply matrices in $\bigO(n^\omega)$ ring operations where $\omega < 3$. Strassen's first $\omega$ was $\log_{2} 7 \approx 2.8$, arising from the fact that only 7 recursions on matrix multiplications of half the size were required as opposed to 8. The current best $\omega$ is $\approx 2.37$, via Le Gall \cite{LeGall2014}. Many of these methods are known to be extremely inefficient compared to naive matrix multiplication for small matrix sizes, although this is usually for the floating point case. For the purposes of this paper, we are interested in any such algorithm where $\omega < 3$, and note that analysis for the ideal cross over point for recursion in the symbolic case for specific polynomials, between the na\"{i}ve and Strassen--Winograd \cite{Winograd1971} algorithms was covered in \cite{TonksSankaranDavenport2017}.

Strassen also introduced methodology for inversion of a matrix via matrix multiplication, arising from the following:

\begin{equation} \label{eqn:A} A = \left( \renewcommand\arraystretch{1.5} \begin{array}{cc} \A11 & \A12 \\ \A21 & \A22 \\ \end{array} \right) = \left( \renewcommand\arraystretch{1.5} \begin{array}{cc} I & 0 \\ \A21 \A11^{-1} & I \\ \end{array} \right) \left( \renewcommand\arraystretch{1.5} \begin{array}{cc} \A11 & 0 \\ 0 & \Delta \\ \end{array} \right) \left( \renewcommand\arraystretch{1.5} \begin{array}{cc} I & \A11^{-1} \A12 \\ 0 & I \\ \end{array} \right) \end{equation}

\noindent where $\Delta = \A22 - \A21 \A11^{-1} \A12$, and so

\begin{equation} \label{eqn:BH} A^{-1} = \left( \renewcommand\arraystretch{1.5} \begin{array}{cc} \A11^{-1} + \A11^{-1} \A12 \Delta^{-1} \A21 \A11^{-1} & - \A11^{-1} \A12 \Delta^{-1} \\ -\Delta^{-1} \A21 \A11^{-1} & \Delta^{-1} \\ \end{array} \right) \end{equation}

Bunch \& Hopcroft followed this with a completeness result in 1974 \cite{BunchHopcroft1974}, which covered the case where an $A_{1,1}$ or $\Delta$ received during the algorithm is singular, while the original input matrix $A$ is non singular (and so inversion or taking a determinant should make sense).

Multivariate polynomial matrices appear, for example, in Gr{\"o}bner basis calculations, where we must upper triangulize dense matrices of sparse polynomials. This is essentially the case we wish to cover, and we note that methods involving interpolation will work poorly here. 

\subsection{Fraction Free Inversion}

However, Section \ref{sect:FFMultInversion} assumes a field, as inversion is analogous with fractions - the inverse of a matrix of polynomials is in general a matrix of rational functions. In performing these algorithms symbolically, this means we are forced to take GCDs (Greatest Common Divisors) of potentially large polynomials in order to simplify results of inversion. The worst case complexity of taking the GCD is polynomial in the degree of the polynomials and exponential in the number of variables \cite{Brown1971}. As far as we are concerned the number of variables may as well be fixed, but the degree of numerators and denominators will increase per recursion as a result of the matrix arithmetic performed in any one recursion.

Our original full presentation of a fraction free modification to Strassen's inversion follows, as was given in \cite{TonksSankaranDavenport2017}, where we investigated the ideal cross over point between standard matrix multiplication and Strassen's methodology in the purely symbolic case. We omit lines indicating that the $A_{i,j}$, $B_{i,j}$ are the canonical submatrices of $A$, $B$ respectively. The operators $+$ and $\cdot$ are the operators acting on $\M Rn \times R^\dag$ essentially canonical as per, for example, the rational numbers. \\

\begin{algorithm}[H]
\caption{First fraction free formulation of Strassen inversion of a matrix $A$ over a ring $R$. The matrices $\A11$ and $\Delta$ are required to be non singular at every iteration.}
\label{alg:FFBH}
{\fontsize{10}{15}\selectfont
\begin{algorithmic}
\STATE \textbf{Input:} $(A,d_{1}) \in \M Rn \times R^\dag$
\STATE \textbf{Output:} $(B,d_{2}) \in \M Rn \times R^\dag$ such that $A \cdot B =  B \cdot A = d_{1} d_{2} I$, and as such $B$ is $d_1 Adj(A)$, and $d_2$ is $Det(A)$.
\STATE \textbf{begin algorithm} FFInversion($A,d_{1}$)
\IF{$n = 1$}
\STATE $B \leftarrow (d_{1})$, $d_{2} \leftarrow \A11$
\ELSE
\STATE $(A^{adj}_{1,1},a_{1,1}^{adj}) \leftarrow$ FFInversion($\A11,1$)
\STATE $(\Delta,\delta) \leftarrow (\A22,1) + (-\A21,1) \cdot (A^{adj}_{1,1},a_{1,1}^{adj}) \cdot (\A12,1)$
\STATE $(\Delta^{adj},\delta^{adj}) \leftarrow$ FFInversion($\Delta,\delta$)
\STATE $(\Lambda,\lambda) \leftarrow (\Delta^{adj},\delta_{adj}) \cdot (\A21,1) \cdot (\A11^{adj},a_{1,1}^{adj})$
\STATE $(B_{1,1}^{'},b_{1,1}) \leftarrow (\A11^{adj},a_{1,1}^{adj}) \cdot\bigl[(I^{\frac n2\times\frac n2},1) + (\A12,1) \cdot (\Lambda,\lambda) \bigr]$
\STATE $(B_{1,2}^{'},b_{1,2}) \leftarrow (-A^{adj}_{1,1},a_{1,1}^{adj}) \cdot (\A12,1) \cdot (\Delta^{adj},\delta^{adj})$
\STATE $d_{2} \leftarrow \Det(A)$ 
\STATE $d \leftarrow d_{1} d_{2}$ 
\STATE $B_{1,1} \leftarrow \frac{d}{b_{1,1}} B_{1,1}^{'}$
\STATE $B_{1,2} \leftarrow \frac{d}{b_{1,2}} B_{1,2}^{'}$
\STATE $B_{2,1} \leftarrow \frac{-d}{\lambda} \Lambda$
\STATE $B_{2,2} \leftarrow \frac{d}{\delta^{adj}} \Delta^{adj}$
\ENDIF
\end{algorithmic}
}
\end{algorithm}

Note that Algorithm \ref{alg:FFBH} is the ``direct translation'' of Strassen's inversion process. The main difference is the definition of $\Delta$, which in the classical case is $\A22 - \A21 \A11^{-1} \A12 = \A22 - \frac{\A21 \A11^{adj} \A12}{a_{1,1}^{adj}}$, but here we multiply through by $a_{1,1}^{adj^{2}}$ (where $a_{1,1}^{adj}$ is actually the determinant of $\A11$) to obtain the fraction free object $a_{1,1}^{adj} \A22 - \A21 \A11^{adj} \A12$.  When doing fraction free inversion, it makes most sense to talk about the pair $(\Adj(A),\Det(A))$, as a result of the fact the intention was to produce inverses, and $\Det(A) A^{-1} = \Adj(A)$. As such when we discuss ``fraction free inverse'', we will largely mean this pair, especially in the context of inverses via the presented methods. \\

\noindent Under this formulation, (\ref{eqn:A}) becomes:

\begin{equation} A = \left( \renewcommand\arraystretch{1.5} \begin{array}{cc} I & 0 \\ \frac{\A21 \A11^{adj}}{a_{1,1}^{adj}} & I \\ \end{array} \right) \left( \renewcommand\arraystretch{1.5} \begin{array}{cc} \A11 & 0 \\ 0 & \Delta \\ \end{array} \right) \left( \renewcommand\arraystretch{1.5} \begin{array}{cc} I & \frac{\A11^{adj} \A12}{a_{1,1}^{adj}} \\ 0 & I \\ \end{array} \right) \end{equation}

\noindent where $\Delta = a_{1,1}^{adj} \A22 - \A21 \A11^{adj} \A12$, $a_{1,1}^{adj} = Det(\A11)$, and $\delta^{adj} = Det(\Delta)$. Then (\ref{eqn:BH}) becomes:

$$ A^{adj} = Det(A) \left( \renewcommand\arraystretch{1.5} \begin{array}{cc} \frac{a_{1,1}^{adj} \delta^{adj} \A11^{adj} +  \A11^{adj} \A12 \Delta^{adj} \A21 \A11^{adj}}{a_{1,1}^{adj^{2}} \delta^{adj}} & \frac{- \A11^{adj} \A12 \Delta^{adj}}{a_{1,1}^{adj} \delta^{adj}} \\ \frac{-\Delta^{adj} \A21 \A11^{adj}}{a_{1,1}^{adj} \delta^{adj}} & \frac{\Delta^{adj}}{\delta^{adj}} \\ \end{array} \right) $$

\begin{equation} = \left( \renewcommand\arraystretch{1.5} \begin{array}{cc} a_{1,1}^{adj^{\frac{n}{2} - 3}} (a_{1,1}^{adj} \delta^{adj} \A11^{adj} +  \A11^{adj} \A12 \Delta^{adj} \A21 \A11^{adj}) & - a_{1,1}^{adj^{\frac{n}{2} - 2}} \A11^{adj} \A12 \Delta^{adj} \\ - a_{1,1}^{adj^{\frac{n}{2} - 2}} \Delta^{adj} \A21 \A11^{adj} & a_{1,1}^{adj^{\frac{n}{2} - 1}} \Delta^{adj} \\ \end{array} \right) \end{equation}

\noindent via the fact discussed in Section \ref{section:improvements} immediately below.

\subsubsection{Immediate Improvements \label{section:improvements}}

As postulated in \cite{TonksSankaranDavenport2017}, we have that $\Det(A)$ arises naturally in the course of computation as $\frac{\delta^{adj}}{(a_{1,1}^{adj})^{\frac{n}{2}-1}}$. Immediately, this saves on computations used to calculate $\Det(A)$, which could previously be obtained in $\bigO(n^{\omega})$ operations via the fraction free formula $\Det(\frac{\Delta}{a_{1,1}^{adj}})\Det(A_{1,1}^{adj})$, which lends itself to a recursive definition requiring adjugates of submatrices, and hence mutually dependent on Algorithm \ref{alg:FFBH}. Removing the need to calculate this via any method removes a significant amount of work, but only up to a constant in terms of the $\bigO$-complexity of the algorithm.

However this formula actually does more - by substituting $\frac{\delta^{adj}}{(a_{1,1}^{adj})^{\frac{n}{2}-1}}$ for $d_{2}$ in all subsequent steps of Algorithm \ref{alg:FFBH}, and working out exact expressions for the pairs calculated via the new operators, we see that significant cancellations occur in calculating the $B_{i,j}^{'}$, and in fact we can obtain much more concise expressions for these matrices and polynomials. In fact, the full definitions of the operators $\cdot$ and $+$ on pairs $(M,c)$ in $R^{n \times n} \times R^{\dag}$ begin to seem superfluous, and in fact we may as well calculate every object directly. Throughout the course of the next section, we see that even the $b_{i,j}$ become redundant, as we require less exact divisions.

We may as well not recurse via this type of method all the way down to size 1. If $A$ is a $2 \times 2$ matrix, the well known method for computing the adjugate is a matter of negating or swapping submatrices of $A$ (the cofactor matrix), and the determinant can be computed in 2 multiplications and an addition (of size 1 matrices ie. polynomials). In contrast, computing the equivalent pair at this dimension via Algorithm \ref{alg:FFBH} can already be seen to require at least as many operations just in computing the relevant $\Delta$ and $\Lambda$ (4 multiplications and an addition).

With a view to these improvements, Algorithm \ref{alg:FFBH2} is a superior implementation that considers the above. Note that the claim that $d_2$ is $Det(A)$ is inductive on $n$.

\begin{algorithm}[H]
\caption{Second fraction free formulation of Strassen inversion of a matrix $A$ over a ring $R$. The matrices $\A11$ and $\Delta$ are required to be non singular at every iteration.}
\label{alg:FFBH2}
{\fontsize{10}{15}\selectfont
\begin{algorithmic}
\STATE \textbf{Input:} $(A,d_{1}) \in \M Rn \times R^\dag$
\STATE \textbf{Output:} $(B,d_{2}) \in \M Rn \times R^\dag$ such that $A \cdot B =  B \cdot A = d_{1} d_{2} I$, and as such $B$ is $d_1 Adj(A)$, and $d_2$ is $Det(A)$.
\STATE \textbf{begin algorithm} FFInversion2($A,d_{1}$)
\IF{$n = 1$}
\STATE $B \leftarrow (d_{1})$, $d_{2} \leftarrow \A11$
\ELSIF{$n=2$}
\STATE $B \leftarrow  \left( \begin{array}{cc} \A22 & -\A12 \\ -\A21 & \A11 \\ \end{array} \right) $, $d_{2} \leftarrow \A11 \A22 - \A12 \A21 $
\ELSE
\STATE $(A^{adj}_{1,1},a_{1,1}^{adj}) \leftarrow$ FFInversion2($\A11,1$)
\STATE $\Delta \leftarrow a_{1,1}^{adj} \A22 -\A21 A^{adj}_{1,1} \A12$
\STATE $(\Delta^{adj},\delta^{adj}) \leftarrow$ FFInversion2($\Delta,a_{1,1}^{adj}$)
\STATE $d_{2} \leftarrow \frac{\delta^{adj}}{(a_{1,1}^{adj})^{\frac{n}{2}-1}}$
\STATE $\Delta^{adj^{'}} \leftarrow \frac{1}{(a_{1,1}^{adj})^{\frac{n}{2}-1}} \Delta^{adj} $
\STATE $\Lambda \leftarrow \frac{1}{a_{1,1}^{adj}} \Delta^{adj^{'}} \A21 \A11^{adj}$
\STATE $\epsilon \leftarrow \A11^{adj} \A12$
\STATE $B_{1,1} \leftarrow \frac{d_{1}}{a_{1,1}^{adj}} (d_{2} \A11^{adj} + \epsilon \Lambda) $
\STATE $B_{1,2} \leftarrow \frac{-d_{1}}{a_{1,1}^{adj}} (\epsilon \Delta^{adj^{'}} )$
\STATE $B_{2,1} \leftarrow -d_{1} \Lambda$
\STATE $B_{2,2} \leftarrow d_{1} \Delta^{adj^{'}}$
\ENDIF
\end{algorithmic}
}
\end{algorithm}

\section{Analysis of Algorithm \ref{alg:FFBH2}}

We consider Algorithm \ref{alg:FFBH2} from here onwards. In particular this means we recurse down to size 2, and thus have one less recursion level. So we never obtain $1 \times 1$ ``matrices'', which is a somewhat trivial case that would have consequences with regards to ``matrix content'', which will be defined in due course. Note substitution of $\Det(A)$ has resulted in cancellations, which results in fewer required divisions (especially on the $B_{i,j}$).

\subsection{Degree Bloat in Intermediate Matrices\label{sect:degbloat}}

Note that if $A$ is a matrix of size $n$ and degree $d$, then $\Adj(A)$ will be a matrix of degree $(n-1)d$, and $\Det(A)$ will be a single polynomial of degree $nd$, as they are sums of $n-1$ and $n$ products respectively. Informally, this implies that working with adjugates means that we will introduce matrices of at least degree $\bigO(nd)$, but we should avoid performing arithmetic on them, possibly by performing cancellations where possible. Indeed our fraction free formulation relies on the multiplication, and introducing multiplications of polynomials of degree anything other than $d$ invalidates any assumption that the cost of multiplication is fixed (and hence the ratio of cost of multiplication to addition, which was a large focus of our previous work).

First, note that if we say a matrix is degree $\bigO(d)$, then we mean all its elements are polynomials of degree $\bigO(d)$. Further, we must make rigorous a concept analogous to ``content'' for polynomials viewed recursively over their variables. For the following two definitions, assume $C$ to be the result of some arithmetic and/or inversion operations from the submatrices of an original matrix $A$, where $C$ is of a binary size possibly less than $n$, but more than 1 - it makes no sense to talk about the content of a matrix of size 1, as this would be trivial, as every element of the matrix is divisible by itself. This $A$ should be assumed to be ``random'' in some sense, or ideally ``symbolic'', which will be defined in due course. For the purposes of the following Definition \ref{defn:MC}, when we say $C_{i,j}$, we mean the $i,j$th element of $C$, rather than the canonical submatrix that we usually mean throughout.

\begin{defn} \label{defn:MC} Define the \emph{content} of a matrix to be the gcd of all its elements, analogous to the same concept for a polynomial. ie. $$content_{gcd}(C) = gcd(C_{1,1},gcd(C_{1,2},gcd(...,gcd(C_{2,1},gcd(...,gcd(C_{n,n-1},C_{n,n}))))))$$ \end{defn}

\begin{defn} \label{defn:SMC} Define the \emph{systematic matrix content}, $content_{sys}(C)$ of a matrix to be the content that we can \emph{deduce} the matrix to have. That is, if $C$ is a matrix that is the result of some arithmetic and inversion operations from some symbolic matrix $A$ (such that every entry of $A$ is a distinct symbol $a_{i,j}$), then $c$ is the systematic matrix content of $C$ if $c$ is the polynomial of largest degree such that every element of $C$ is divisible by $c$. In this way, the systematic matrix content of a symbolic matrix is equal to $content_{gcd}(C)$. \end{defn}

The distinction between the above definitions is important. It is clear that \par \noindent $content_{sys}(C) \mid content_{gcd}(C)$. When discussing deterministic algorithms, Definition \ref{defn:SMC} allows us to talk of factors that we know will always be there in any one intermediate matrix. Content in the context of Definition \ref{defn:MC} somewhat implies that we must take pairwise GCDs of elements in $A$ to compute this object. This content may be more than we bargained for, and depending on the structure of the original matrix $A$, we may deduce ``accidental'' factors that wouldn't be there for every starting matrix $A$.

In many scenarios, $content_{sys}(C) = content_{gcd}(C)$. One situation where this is not the case is when $A$ had content to begin with, ie. $degree(content(A)) > 0$. Then there is a discrepancy between $content_{sys}(A)$ and $content_{gcd}(A)$ that will manifest in, for example the matrix $\Delta$, which is a non linear combination of elements in $A$. In any case, whenever we refer to ``content'' further, we mean systematic matrix content in the context of Definition \ref{defn:SMC}, and so take $content = content_{sys}$. If we mean content in terms of Definition \ref{defn:MC}, we will use $content_{gcd}$.

Note that for some polynomial $c$ and some matrix $A$, $\Adj(c A) = c^{n-1} \Adj(A)$. This is especially important if we are to view $c$ in terms of matrix content of $A$ - we can, at the very least, calculate $\Adj(A)$ from a ``primitive'' $A$, and as such multiply only the polynomials of lowest degree possible in the context of matrix multiplication in the process of Algorithm \ref{alg:FFBH2}. This is the motivation behind Definition \ref{defn:SMC}. In fact, it will be seen further that the term $c^{n-1}$ is quite often the factor that gets divided through (as can be seen frequently in Algorithm \ref{alg:FFBH}), and so in practice we can discard it.

\subsection{Analysis of Degree Bloat in Intermediate Matrices\label{sect:deganalysis}}

Forward, we assume that $d_2$ from Algorithm \ref{alg:FFBH2} is $Det(A)$.

\begin{remark} Where we claim a matrix to have content, this can be verified by computation, for example in Maple, from an original symbolic matrix. Thus, we defer proof of claims of matrix content to computation. We note that the deterministic nature of all computation involved ensures that in the context of systematic matrix context, this will always hold for that particular matrix. A matrix divided through by its content will always result in an integral matrix, ie. fraction free / the divisions are exact. \end{remark}


\begin{thm} \label{thm:toplevelO} All intermediate matrices at the top level for Algorithm \ref{alg:FFBH2} are of degree $\bigO(nd)$. \end{thm}

\begin{proof} We assume $n \geq 4$ such that $\Delta$ makes sense in the context of Algorithm \ref{alg:FFBH2}. The fraction free inversion of $(\A11,1)$ returns $(\A11^{adj},a_{1,1}^{adj}) = (\Adj(\A11),\Det(\A11))$. $\A11^{adj}$ is of degree $d(\frac{n}{2}-1)$ as $\A11$ is of size $\frac{n}{2}$, and similarly $a_{1,1}^{adj}$ is a single polynomial of degree $\frac{nd}{2}$. $\Delta$ is $a_{1,1}^{adj} \A22 - \A21 \A11^{adj} \A12$, which is of degree $d(\frac{n}{2} + 1)$, and $\delta = a_{1,1}^{adj}$. We then calculate the fraction free inverse of $(\Delta,\delta)$ as $(\Delta^{adj},\delta^{adj})$. $\Delta^{adj}$ is the adjugate of $\Delta$, multiplied through by a factor of $a_{1,1}^{adj}$ (as $\delta = a_{1,1}^{adj}$ appears as $d_{1}$ in Algorithm \ref{alg:FFBH2} here), so $\Delta^{adj}$ has degree $d(\frac{n}{2}+1)(\frac{n}{2}-1) + \frac{nd}{2}$. Note that now we divide through by $a_{1,1}^{adj^{m-1}}$, which is degree $\frac{n^{2} d}{4} - \frac{nd}{2}$. Thus $\Delta^{adj}$ ends up being degree $(n-1)d$ (the degree of the adjugate). All ensuing operations at the top level is matrix arithmetic on matrices of degree $\bigO(nd)$, and as such all resulting matrices are $\bigO(nd)$. Note that the cancellations involving $a_{1,1}^{adj}$ are $\bigO(nd)$ cancellations, and so only end up subtracting a factor of $nd$ up to a constant in the degrees of the matrices involved. \end{proof}

Note that additionally in this formulation, we can cancel an extra factor of $a_{1,1}^{adj}$ from $\Lambda$, such that $\Lambda$ becomes degree $d(n-1)$. \\

The above demonstrates that adjugates and $\Delta$ are both inherently ``degree $nd$'' objects, in that calculating them \emph{will} introduce a factor of $n$ to the degrees of the polynomials in arithmetic. So the ratio of cost of addition to multiplication certainly isn't fixed. However in particular, Algorithm \ref{alg:FFBH2} being inherently recursive should exacerbate this issue, as we must take ``the $\Delta$ of $\Delta$'', and ``the adjugate of $\Delta$'' and so on.

At first, one should be most concerned about the chain of $\Delta$s that ensues upon calculation of inversion of a suitably large matrix $A$, which we shall investigate as a case study into a special case of the bigger picture.

\begin{thm} \label{thm:delta} Let $A \in \M Rn$. Let $\Delta_{0}$ denote the $\Delta$ calculated from the submatrices of $A$, and for all $1 \leq k \leq \log_{2} n - 1$, let $\Delta_{k}$ denote the $\Delta$ calculated from the submatrices of $\Delta_{k-1}$. Then $\Delta_{k}$ is of degree $\bigO(n^{k+1} d)$ for all $0 \leq k \leq \log_{2} n - 1$. \end{thm}

\begin{proof} Note that $0 \leq k \leq \log_{2} n - 1$ ensures that $\Delta_{k}$ actually exists in the context of Algorithm \ref{alg:FFBH2}. $\Delta_{\log_{2} n -1}$ is of size 4, and as such the last $\Delta$ to fully be inverted by recursion for Algorithm \ref{alg:FFBH2}. $\Delta_{0}$ is of degree $d(\frac{n}{2} +1) \in \bigO(nd)$ by Theorem \ref{thm:toplevelO}. Now assume that $\Delta_{i}$ is of degree $\bigO(n^{i+1} d)$ for some $0 \leq i \leq \log_{2} n - 1$. We have $\Delta_{i_{1,1}}$ is of degree $\bigO(n^{i+1} d)$, and so the associated determinant $\delta_{i_{1,1}}^{adj}$ is of degree $\bigO(n^{i+2} d)$. $\Delta_{i_{1,1}}^{adj}$ is similarly of degree $\bigO(n^{i+2} d)$, and $\Delta_{i_{1,2}} , \Delta_{i_{2,1}} , \Delta_{i_{2,2}}$ are all of degree $\bigO(n^{i+1} d)$. Hence the object $\Delta_{i+1} = \delta_{i_{1,1}}^{adj} \Delta_{i_{2,2}} - \Delta_{i_{2,1}} \Delta_{i_{1,1}}^{adj} \Delta_{i_{1,2}}$ is of degree $\bigO(n^{i+2} d)$. By induction, $\Delta_{k}$ is of degree $\bigO(n^{k+1} d)$ for all $ 1 \leq k \leq \log_{2} n - 1$. \end{proof}

\begin{remark} In recursions on Algorithm \ref{alg:FFBH2} on $A$, for matrices of size $\frac{n}{2^{i}}, i \geq 1$ there are $2^i$ matrices that are the results of taking various `$\A11$' and `$\Delta$' operations that we must invert, however the worst case really is inversions involving as many `$\Delta$' operations as possible, as per Theorem \ref{thm:delta}. In contrast, if we were to talk about `$A_{1,1_{k}}$' in the same manner as Theorem \ref{thm:delta}, then $A_{1,1_{k}}$ would be of degree $d$ for all $k \geq 1$. Section \ref{section:mixedcases} discusses the degree of ``mixed cases'', which is the generalisation. \end{remark}

Further, when we say ``iteration $k$'', we mean all inversions on matrices of size $\frac{n}{2^{k}}$, but in particular this recursion level involves $\Delta_{k-1}$ from Theorem \ref{thm:delta}, and this matrix is certainly the worst case in terms of degree. The following Theorem \ref{thm:arithmetic} generalises Theorem \ref{thm:toplevelO}.

\begin{thm} \label{thm:arithmetic} The matrices in arithmetic in iteration $k$ of Algorithm \ref{alg:FFBH2} (on $\Delta_{k-1}$) are of degree $\bigO(n^{k+1} d)$. \end{thm}

\begin{proof} At iteration $k$ on $\Delta$, we calculate $\Delta_{k}$, which via Theorem \ref{thm:delta}  is of degree $\bigO(n^{k+1} d)$. Inversion of this gives us $\Delta_{k}^{adj}$, but in fact we cancel it to receive $\Delta_{k}^{adj^{'}}$, and in a similar manner to in Theorem \ref{thm:toplevelO} this is of degree $\bigO(n^{k+1} d)$. All further arithmetic at this level is on matrices of degree $\bigO(n^{k+1} d)$ and thus leads to matrices of degree $\bigO(n^{k+1} d)$. \end{proof}

\begin{thm} \label{thm:endresult} Fraction free inversion via Algorithm \ref{alg:FFBH2} on a matrix of size $n$ and degree $d$ involves arithmetic on matrices of degree $\bigO(n^{{\log_2 n} - 1} d)$. \end{thm}

\begin{proof} Algorithm \ref{alg:FFBH2} calls itself recursively on $\Delta_{i}$ matrices (ie. we are in the ``else'' case of Algorithm \ref{alg:FFBH2}) ${\log_2 n}  - 2$ times if we stop recursion on Algorithm \ref{alg:FFBH2} at size 2 and calculate adjugates via the transpose of the cofactor matrix. Using Theorem \ref{thm:arithmetic}, we can deduce that the matrices in arithmetic at size 2 (ie. recursion level $\log_{2} n -2$) are of degree $\bigO(n^{\log_{2} n -1} d)$. \end{proof}

\begin{remark} \label{rmk:endremark} The equivalent results from Theorems \ref{thm:arithmetic} and \ref{thm:endresult} for Algorithm \ref{alg:FFBH} would have the degrees at $\bigO(n^{k+2} d)$ and $\bigO(n^{\log_{2} n + 1}d)$ respectively.  \end{remark}

\begin{proof} Consider Algorithm \ref{alg:FFBH}, and observe that Theorem \ref{thm:delta} holds the same in this case. In Theorem \ref{thm:arithmetic}, we do not perform cancellation on $\Delta_{k}^{adj}$, and so as the adjugate of $\Delta_{k}$ which is of degree $\bigO(n^{k+1} d)$, it is of degree $\bigO(n^{k+2} d)$, and then all arithmetic is now on matrices of this degree. For Algorithm \ref{alg:FFBH} we instead perform $\log_{2} n - 1$ recursions due to recursing down to matrices of size 1. In total, the equivalent result for Theorem \ref{thm:endresult} has that we perform arithmetic on matrices of degree $\bigO(n^{\log_{2} n + 1} d)$.  \end{proof}

\subsection{Cancellations in Intermediate Matrices}

Having to perform arithmetic on matrices of degree $\bigO(n^{{\log_2 n} - 1} d)$ is beyond the realms of being competitive - however in this section we discuss what cancellations can be provided to mitigate this degree increase. This requires some discussion of mixed cases, which so far we have neglected. But we first present the ``pure'' $\Delta$ case, ie. always choosing the right subtree from Figure \ref{fig:binarytree}.

\begin{remark} Assume from here onwards that when computing any intermediate matrix, we do so from a cancelled down matrix, ie. from a matrix with a systematic matrix content of 1. In the context of Section \ref{sect:degbloat}, $\Delta_{k}$ would be calculated from a $\Delta_{k-1}$ that has been divided through elementwise by the content defined in due course. \end{remark}

\begin{ntn} \label{ntn:submatrixntn} The notation $A_{1..i,1..j}$ where $i,j \in \{1,..,n\}$ means the upper left submatrix with $i$ rows and $j$ columns. As an example, $\A11$ would be $A_{1..\frac{n}{2},1..\frac{n}{2}}$ in this notation. Additionally, let $A_{1..i,1..j}$ be 1 for $i,j < 1$. \end{ntn}

\begin{thm} \label{thm:deltaicontents}If $ 0 \leq i \leq \log_{2}(n) - 2$, $\Delta_{i}$ has content $\Det(A_{1..\frac{(2^{i}-1)n}{2^{i}},1..\frac{(2^{i}-1)n}{2^{i}}})^{\frac{n}{2^{i+1}}}$. After cancelling through by this content, $\Delta_{i}$ is of degree $(\frac{(2^{i+1}-1)n}{2^{i+1}}+1)d \in \bigO(nd)$, for all $ 0 \leq i \leq \log_{2}(n) - 2$. \end{thm}

\begin{proof} Note that $0 \leq i \leq \log_{2}(n) - 2$ ensures that $\Delta_{i}$ is a matrix of size at least 2, such that the definition of matrix content makes sense.

As above, we calculate each $\Delta_{i}$ from a $\Delta_{i-1}$ that has been divided through by the content from \ref{thm:deltaicontents} and thus has a resulting content of 1. The theorem is true for $i=0$ as by computation on a symbolic matrix of any size we can verify it has a content of 1, which is equivalent to $A_{1..0,1..0}$ under Notation \ref{ntn:submatrixntn}, and Theorem \ref{thm:delta} has that $\Delta_{0}$ is of degree $d(\frac{n}{2}+1)$.

Now assume $\Delta_{i}$ is of degree $(\frac{(2^{i+1}-1)n}{2^{i+1}}+1)d$ for some $i \geq 0$. Then $\Delta_{i}$ is of size $\frac{n}{2^{i+1}}$, and the $\A11$ for $\Delta_{i}$ has the same degree as $\Delta_{i}$, and is of size $\frac{n}{2^{i+2}}$. So $\Det(\A11)$ is of degree $\frac{n}{2^{i+2}}(\frac{(2^{i+1}-1)n}{2^{i+1}}+1)d$, and so $\Delta_{i+1}$ is of degree $(\frac{n}{2^{i+2}}+1)(\frac{(2^{i+1}-1)n}{2^{i+1}}+1)d$. Expanding this, we get $(\frac{(2^{i+1}-1)n^{2}}{2^{2i+3}}+ \frac{(2^{i+2}-1)n}{2^{i+2}}+1)d$. We know via Theorem \ref{thm:deltaicontents} that we can divide through by $\Det(A_{1..\frac{(2^{i+1}-1)n}{2^{i+1}},1..\frac{(2^{i+1}-1)n}{2^{i+1}}})^{\frac{n}{2^{i+2}}}$, which is of degree $\frac{(2^{i+1}-1)n^{2}d}{2^{2i+3}}$. As such after cancellation $\Delta_{i+1}$ is of degree $(\frac{2^{i+2}-1)n}{2^{i+2}} + 1)d$,  which is precisely the induction statement but with $i$ incremented by 1, evidently $\Delta_{i+1}$ is of degree $\bigO(nd)$ as a result, and thus by mathematical induction the statement holds for all $i \geq 1$. \end{proof}

\begin{thm}\label{thm:deltaadjicontent} If $0 \leq i \leq \log_{2}(n) - 3$,  $\Delta_{i}^{adj}$, when computed from a corresponding $\Delta_{i}$ of matrix content 1, has content $\Det(A_{1..\frac{(2^{i+1}-1)n}{2^{i+1}},1..\frac{(2^{i+1}-1)n}{2^{i+1}}})^{\frac{n}{2^{i+1}}-2}$. After dividing through by this content, $\Delta^{adj}_{i}$ achieves degree $\bigO(nd)$.\end{thm}

\begin{proof} Note that $0 \leq i \leq \log_{2}(n) - 3$ ensures that $\Delta^{adj}_{i}$ is a matrix of size at least 4. $\Delta^{adj}_{\log_{2} n - 2}$, ie. the $\Delta^{adj}_{i}$ of size 2, is the result of transposition of the cofactor matrix of $\Delta_{i}$, and as such no arithmetic is performed, and as such its content is certainly 1.

Using Theorem \ref{thm:deltaicontents}, we have that after appropriate cancellations $\Delta_{i}$ is of degree $(\frac{(2^{i+1}-1)n}{2^{i+1}}+1)d$, and size $\frac{n}{2^{i+1}}$. Computing $\Delta_{i}^{adj}$ we get that it is of degree $(\frac{n}{2^{i+1}}-1)(\frac{(2^{i+1}-1)n}{2^{i+1}}+1)d$. Expanding this gives $(\frac{(2^{i+1}-1)n^{2}}{2^{2i+2}} +\frac{(2-2^{i+1})n}{2^{i+1}} -1)d$. Now via Theorem \ref{thm:deltaadjicontent}, we can divide through by $\Det(A_{1..\frac{(2^{i+1}-1)n}{2^{i+1}},1..\frac{(2^{i+1}-1)n}{2^{i+1}}})^{\frac{n}{2^{i+1}}-2}$, which is of degree $(\frac{(2^{i+1}-1)n^{2}}{2^{2i+2}} + \frac{(2-2^{i+2})n}{2^{i+1}})d$. Hence after cancellation $\Delta_{i}^{adj}$ is of degree $(\frac{(2^{i+2}-2^{i+1})n -1}{2^{i+1}})d = (n-1)d \in \bigO(nd)$. \end{proof}

\subsection{Mixed Cases\label{section:mixedcases}}

\begin{figure}[h]
\centering 
\includegraphics[width=12cm]{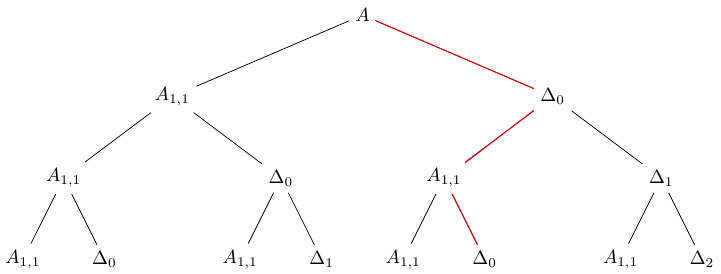}
\caption{A representation of the binary tree for the inversion of a matrix $A$ (that should be of size at least 16 here) via Algorithm \ref{alg:FFBH2}. We simplify notation such that ``the $\A11$ of $\Delta_{0}$'' is just labeled as ``$\A11$''. Note that the tree is left-oriented, that is for any node, the inversion of the left subtree is required before the inversion of the right subtree can be computed, which is required for the inversion of the node (computation of $\Delta$ requires the pair $(\A11^{adj},a_{1,1}^{adj})$). The red path shows an example of a ``mixed case''.} \label{fig:binarytree}
\end{figure}

The red path in figure \ref{fig:binarytree} is an example of a ``mixed case''. In ad-hoc functional notation, the red path gives us $M = \Delta(\A11(\Delta(A)))$. Note that with this notation, $\Delta_k$ from \ref{sect:degbloat} becomes $\Delta^k(A)$, and from here onwards we omit the brackets. Essentially, we require the adjugate of $\A11$ in order to invert $\Delta_{0}$, which is required in order to invert $A$. The case study traversed in Section \ref{sect:degbloat} would have us believe that we should treat the inversion of this $\A11$ as if it were a completely new matrix - indeed it has a trivial systematic matrix content as per the first $\Delta_{0}$ (see Theorem \ref{thm:intermediatecontent}), but then this leads us to believe that $M$ has no systematic content (as it is a ``$\Delta_0$''). Except $M$ is the result of two \emph{net} $\Delta$ operations, and thus is of degree $\bigO(n^2 d)$. We want $\bigO(nd)$. \\ 

\begin{thm} \label{thm:intermediatecontent} If $M$ is a matrix produced directly from a recent $\A11$ operation, then $content_{sys}(M) = 1$. Otherwise, keeping in mind that we calculate each matrix from matrices of respective systematic matrix content 1, if $M = \Delta \phi_1 \dots \phi_k A$ where $\phi_j \in \{ A_{1,1}, \Delta \}$ for all $1 \leq j \leq k$, and $A$ is a matrix of binary size $n$ at least $2^{k+2}$, then $content_{sys} (M) = Det(A_{1..\frac{(2^{i} -1) n}{2^{i}},1..\frac{(2^{i} -1)n}{2^{i}}})^{\frac{n}{2^{k+1}}}$, where $i$ is the number of $\phi_j$ that are $\Delta$. Upon cancelling through by this, we achieve degree at most $\bigO(nd)$. \end{thm}

\begin{proof}

If $M$ is a matrix produced directly from a recent $\A11$ operation (ie. $\phi_1 = \A11$), then we are looking at an unmodified submatrix of a matrix that is assumed to have been cancelled through, and so we can deduce no further content, so $content_{sys}(M) = 1$. Hence it suffices to work with sequences of operations beginning with $\Delta$, as we do further. \\

To prove the claim, we work with generic intermediate matrices of the form $M = \Delta \A11^{m_{J}} \Delta \A11^{m_{J-1}} \Delta \dots \Delta \A11^{m_{1}} \Delta A$, where $J>0$, each $m_i \geq 0$ for all $1 \leq i \leq J$, and $A$ is a matrix of size at least $2^{J + 1 + \sumtoJ}$. In the case where $m_i =0$ for all $1 \leq i \leq J$, this is identical to the case of Theorem \ref{thm:deltaicontents}. We induct on $J$ over $\mathds{N}_{0}$. \\

For $J=0$, we obtain $\Delta A$, and this is of degree $(\frac{n}{2} +1)d$, which is of degree $\bigO(nd)$. Now, for the purposes of induction, assume $J \geq 0$, and $M = \Delta A_{1,1}^{m_J} \Delta \dots \Delta A_{1,1}^{m_{1}} \Delta A$ is of degree $$\frac{(1+\sum_{j=1}^{J} 2^{J + 1 - j + \sum_{i=j}^{J} m_i } ) n d}{2^{J + 1 + \sumtoJ}} + 1 - \bigO(n^{k_0} d),$$ for some $k_0 \geq 2$ and of size $2^{- J - 1 -\sumtoJ} n$. Then consider $\Delta \A11^{m_{J+1}} M$, where $m_{J+1} \geq 0$. This is of size $2^{- J - 1 -\sumtoJplusone}n$, and hence of degree $$\left( \frac{n}{2^{J + 2 + \sumtoJplusone}} + 1 \right) \left( \frac{(1+\sum_{j=1}^{J} 2^{J + 1 - j + \sum_{i=j}^{J} m_i} ) n }{2^{J + 1 + \sumtoJ}}+1 \right) d.$$ The coefficient of $n^2 d$ is $$\frac{1 + \sum_{j=1}^{J} 2^{j + \sum_{i=j}^{J} m_i} }{2^{2J + 3 + \sumtoJplusone + \sumtoJ}},$$ and the theorem suggests we can cancel through by $$Det(A_{1..\frac{(2^{J+1} -1) n}{2^{J+1}},1..\frac{(2^{J+1} -1)n}{2^{J+1}}})^{n 2^{-J - 2 -\sumtoJplusone}},$$ which is of degree $$\frac{(2^{J+1} - 1)n^2 d}{2^{\sumtoJplusone + 2J + 3}}.$$ After cancellation, the coefficient of $n^2 d$ becomes:

$$ \frac{1 + \sum_{j=1}^{J} 2^{\sum_{i=j}^{J} m_i + j} - 2^{J+1 + \sumtoJ} + 2^{\sumtoJ}}{2^{2J + 3 +\sumtoJplusone}} $$

\noindent If we are to achieve degree $\bigO(nd)$, this must be non positive.

$$\Rightarrow 1 + \sum_{j=1}^{J} 2^{\sum_{i=j}^{J} m_i + j} - 2^{J+1 + \sumtoJ} + 2^{\sumtoJ} \leq 0$$

\noindent Manipulating the left hand side, we receive:

$$ \leq 1 + \sum_{j=1}^{J} 2^{\sumtoJ + j} - 2^{J+1 + \sumtoJ} + 2^{\sumtoJ}$$

$$ = 1 + 2^{\sumtoJ} ( \sum_{j=1}^{J} 2^{j} - 2^{J+1} + 1)$$

\noindent Using that $\sum_{j=1}^{J} 2^{j} = 2^{J+1} - 2$,

$$ =  1 + 2^{\sumtoJ} ( 2^{J+1} - 2 - 2^{J+1} + 1 ) $$

$$ =  1 - 2^{\sumtoJ} $$

\noindent which is certainly non positive as $m_i \geq 0$ for all $1 \leq i \leq J+1$. Hence the $n^2 d$ term is at most 0, and thus we achieve $\bigO(nd)$. We verify that we match the coefficient of $nd$ from the inductive step, which is:

$$\frac{1}{2^{\sumtoJ + j}} + \frac{1+\sum_{j=1}^{J} 2^{J + 1 - j + \sum_{i=j}^{J} m_i }}{2^{J + 1  + \sumtoJ}} + 1$$

$$ = \frac{1 + 2^{m_{J+1}} + \sum_{j=1}^{J} 2^{\sum_{i=j}^{J+1} m_i + j + 1}}{2^{J + 2 + \sumtoJplusone}} + 1$$

$$ = \frac{1 + \sum_{j=1}^{J+1} 2^{\sum_{i=j}^{J+1} m_i + j}}{2^{J + 2 + \sumtoJplusone}}$$

\noindent which is the induction statement about the degree, but with $J+1$. Hence by induction $M$ achieves degree $\bigO(nd)$ for all $J \geq 0$, and so all mixed case $\Delta$ achieve degree $\bigO(nd)$.

\end{proof}

\begin{remark} Theorem \ref{thm:deltaicontents} is consistent with Theorem \ref{thm:intermediatecontent} by taking or $m_i = 0$ for all $i$ in the latter to get the former, NOT by taking $J=0$. It is important to note $J$ is commensurate with the net number of $\Delta$ operations undertaken, where as the $m_i$ represent the number of $\A11$ operations undertaken. Note that we know from Theorem \ref{thm:deltaicontents} that the cancellations induced in the pure $\Delta$ case completely coincide with the $\bigO(n^2 d)$ term from the degree, ie. the `$\leq$' from the above theorem is actually an `$=$' in this case. \end{remark}

\begin{remark} We explain the presence of $k_0$ in the above proof. Note that in the case that the cancellation ``overshoots'', ie. the `$\leq$' is really a `$<$`, this results in a negative $\bigO(n^2 d)$ term in the degree. Further, this will actually increase by a factor of $n$ every inductive step (every time $J$ increments with an associated $m_J > 0$, ie. there has been an additional tangible chain of $\A11$ operations) by virtue of the multiplication by the term induced by taking the last $\Delta$ operation. As such, a mixed case may actually be of a lower degree than the pure $\Delta$ case of the same size by a term of $\bigO(n^{k_0} d)$ where $k_0$ is proportional to the number of non-zero $m_i$. As above, Theorem \ref{thm:deltaicontents} implies no such increasing negative $\bigO(n^k d)$ term for the pure $\Delta$ case. \end{remark}


\begin{remark} 16 is the first matrix size where we obtain a mixed case $\Delta$ of degree an order higher than $\bigO(nd)$ that we then cancel. If $A$ is a matrix of size 16, this admits the sequence of operations $\Delta \A11 \Delta A$, which has two separate contiguous sequences of $\Delta$s. This results in a matrix of degree $\bigO(n^2 d)$ (precisely, degree $(\frac{n^2}{8} + \frac{3n}{4} + 1)d = 45d$), and theorem \ref{thm:intermediatecontent} suggests we can cancel through by a factor of degree $\frac{n^2}{16}$. A Maple worksheet examining this case can be found at \cite{FastMatrixInversionWorksheet18}. \end{remark}

The equivalent discussion for cancellation of mixed case adjugates (ie. a ``mixed case'' $\Delta^{adj}$) is not traversed here. The existence of cancellations for mixed case $\Delta$ and ``pure'' $\Delta^{adj}$ suggests that cancellations for mixed case $\Delta^{adj}$ are likely, but the formula and degree of this cancellation are both unknown, too whether it results in a cancellation down to $\bigO(nd)$.

\subsection{Limitations of Cancellations \label{sect:limitations}}

As the degree of both elements of the intended output of a fraction free inversion algorithm is $\bigO(nd)$, it is good news that it appears so far that we can work with elements of optimal degree after cancellation. However one has to think about the cancellations themselves, and that requiring these cancellations induces further problems.

The above theorems describe factorisations possible in the binary case - that is the original matrix $A$ was of binary size, and subsequently all intermediate matrices are of binary size. But the factors found are determinants of submatrices of the original $A$, and these submatrices are (not usually) of binary size. In some sense, it would be elegant and perhaps convenient to recurse on Algorithm \ref{alg:FFBH2} to acquire these determinants, in some manner similar to that of Bareiss-Dodgson for Gaussian Elimination \cite{Dodgson1866} \cite{Bareiss1968}. This requires an extension of Algorithm \ref{alg:FFBH2} to be able to accept and work with matrices of non binary size. This can be done via padding the original matrix in a manner similar to that for Strassen-like multiplication. And indeed such a padding solution can also be employed to be able to multiply intermediate matrices that are of odd size, but actually a better solution is probably ``peeling'' as presented in \cite{HLJJTT1996} in the context of Strassen-like multiplication. \\

We briefly discuss padding for inversion. If $A$ is of odd size, we introduce a row of 0s, and then a further column of 0s, and add a 1 on the leading diagonal. This conditions the input matrix to be of even size, and so ensures that we can divide the matrix into quarters. Note that this matrix has the same determinant as the input, and we can ``trim'' the resulting adjugate to receive the adjugate of the actual input. But this has interesting implications on the intermediate matrices depending on if we pad ``upper left'' or ``lower right'' - in particular the elements we introduce are of maximum degree 0, and so our choice of padding has interesting consequences on the $\Delta_i$ matrices - in particular the former choice leads to $\Delta_i$ matrices of uniform degree, whereas the latter choice results in $\Delta_i$ that are more diagonal (ie. we get some 0s off the leading diagonal in the far right and bottom left of the matrix), but of potentially higher maximum degree. A Maple worksheet examining the dichotomy of cases can be found at \cite{FastMatrixInversionWorksheet18}. \\

In particular however, with respect to the obtained cancellations we now require up to four recursions on the algorithm per iteration - the two original recursions on $\A11$ and $\Delta$, and up to two further recursions to provide cancellation factors for $\Delta$ and $\Delta^{adj}$. Furthermore, these two latter recursions may be on submatrices that are close to the whole original matrix $A$. In particular, the pure $\Delta_i$ or $\Delta^{adj}_i$ for larger $i$ may require recursions on matrices of size $\approx n$, which isn't in the spirit of the rest of the algorithms that require recursions purely on matrices of half the size - a disaster for worst case complexity.

An alternative to the additional recursions is to instead calculate the content in the spirit of Definition \ref{defn:MC}. That is, we take pairwise GCDs of elements, but we can in fact be more intelligent and terminate taking the GCDs when we reach the degree we expect the content to have. This is a compromise between systematic and non systematic matrix content, and ensures in best terms possible that we don't catch factors that we wouldn't normally deduce to be there. We can also randomise our selection of elements to take GCDs of in the same spirit as calculating content of polynomials. There is a consideration however that in the situation where $A$ was not uniformly of degree $d$, the above theorems on cancellations only provide upper bounds for the degree that we expect.

\begin{thm} Assuming that every $\Delta^{adj}$ matrix admits cancellation (including those arising from mixed cases), and via the above suggestion that we take pairwise GCDs to cancel $\Delta_i$ and $\Delta_{i}^{adj}$, we require taking $\bigO(n)$ GCDs of polynomials of degree at least $\bigO(n^2 d)$ in $m$ variables. \end{thm}

\begin{proof} At iteration $i$, $1 \leq i \leq \log_{2} n - 1$, we have $2^i$ principal submatrices that are the results of taking various $\A11$ and $\Delta$ operations. Half of these are the result of a $\Delta$ operation \emph{at this level}. Exactly one is the result of just one net $\Delta$ operation, while the rest admit cancellation. This results in $2^{i-1} -1$ matrices of a $\Delta$ description that admit cancellation, $1 \leq i \leq \log_{2} n - 1$. Meanwhile all received $\Delta^{adj}$ admit cancellation, except for those at size 2. This makes $2^{i-1}$ $\Delta^{adj}$ that admit cancellation, $1 \leq i \leq \log_{2} n - 2$. Note that we assume each $\Delta$, $\Delta^{adj}$ is calculated from a matrix of degree $\bigO(nd)$, and as such before cancellation is degree $\bigO(n^2 d)$. Also noting that we must take at least one GCD per matrix in order to obtain its content, we take at least $\sum_{i=1}^{\log_{2} n -1} 2^{i-1} -1 + \sum_{i=1}^{\log_{2} n -2} 2^{i-1} = \frac{3}{2} 2^{\log_{2} n - 1} - \log_{2} n - 1 = \frac{3n}{4} - \log_{2} n - 1 \in \bigO(n)$ GCDs. \end{proof}

\begin{remark} Interestingly, the worst case (in terms of number of GCDs) is at most $\bigO(n^2)$ GCDs of polynomials of degree at least $\bigO(n^2 d)$ if we have to take pairwise GCDs of every element in every intermediate matrix to find the respective content. At iteration $i$, each matrix has $\frac{n^2}{2^{2i-2}}$ elements, so the sum is $n^2 (\sum_{i=1}^{\log_{2} n -1} \frac{(2^{i-1} - 1)}{2^{2i-2}} + \sum_{i=1}^{\log_{2} n -2} \frac{1}{2^{i-1}}) = \frac{4}{3}(2n^2 -9n +4) \in \bigO(n^2)$. Anecdotally, the author finds that the content can usually be found in exactly one step, ie. exactly one GCD (of any two elements) coincides with $content_{gcd}(A)$. Therefore, we expect the average case to be at least $\bigO(n)$ GCDs.\end{remark}

\section{Other methods}

We acknowledge a fraction free formulation arising from work summarised \& reviewed in \cite{ShpilkaYehudayoff2010}, regarding arithmetic circuits \& formulae. Informally, Strassen's formulations can be followed in such a manner that we separate the numerator and denominator of each entry of the result during calculation. In particular however this assumes unit cost for arithmetic to be efficient, and so does not avoid the issue we present. In particular one very large GCD must be taken per entry of the result upon completion of all relevant arithmetic. We also note that this formulation attributes a denominator to every entry of the result at all times, whereas our formulation associates a denominator (implicitly or otherwise) to every matrix throughout the algorithm.

Giorgi, Jeannerod and Villard present a method for inversion of polynomial matrices in \cite{Villardetal2003}, but this is for matrices over $K[x]$, ie. polynomials in one variable, where we cover the case of $m>1$. This claims a running time of $\tilde{\bigO}(n^3 d)$, where $\tilde{\bigO}$ is the usual ``soft $\bigO$'' which drops logarithmic terms. Similar claims are made for determinants, but with a running time of $\tilde{\bigO}(n^\omega d)$, via Storjohann \cite{Storjohann2002}. While their solutions do use fast matrix multiplication, the main part of the solution in \cite{Villardetal2003} revolves around ``$\sigma$-bases'', which don't appear to generalise to the case of multiple variables.

One notes that one could use Kronecker substitutions \cite{ArnoldRoche2014} to reduce the case of multivariate polynomial matrices to the case of univariate polynomial matrices, and then use the method discussed above. However, this will increase the degree of polynomials considerably, and it seems unlikely that this outperforms a direct approach.

\section{Conclusions and Further Questions}

Indeed with respect to matrices of multivariate polynomials, a fraction free solution as a translation of Strassen's original inversion certainly exists, but has several caveats. An adaptation of Strassen's inversion requires using an adaptation of the matrix $\Delta$, which in the classical case is in general a matrix of rational functions, but in the fraction free case becomes a degree $\bigO(nd)$ object. This is a result of adjugates / determinants being ``degree $nd$'' objects. These modified $\Delta$ matrices have a significant amount of ``content'', which we find to be redundant. While this content is directly computable via recursion, it is most likely better to remove this content via taking GCDs ``intelligently''. We note that for the univariate case covered in \cite{Villardetal2003}, with each recursion they obtain two matrices of half the size of degree $2d$, where as for our multivariate case we obtain two matrices of half the size of degree at least $\bigO(nd)$ after cancellation, which greatly increases the worst case complexity.  \\

\noindent In total, the following remain unexplored:
\begin{itemize} 
\item The equivalent of Theorem \ref{thm:intermediatecontent} for adjugates (ie. mixed case $\Delta^{adj}_{i}$) is still unknown. The fact that the pure case generalises to the mixed case for $\Delta$ seems to make this hopeful, but it is unknown what the factors will be, nor their degree - one would hope this takes the matrices down to degree $\bigO(nd)$ to match optimal degree as per the rest.
\item What one does to adapt Algorithm \ref{alg:FFBH2} with the cancellations arising from the theorems in this work. In doing so one should essentially require none of the intermediate divisions that presently appear towards getting the $B_{i,j}$, where instead these occur directly on $\Delta$ and $\Delta^{adj}$ in order to ensure the ensuing arithmetic is all on degree $\bigO(nd)$ objects at every iteration. This corresponds to the claim that we find the content to be ``redundant'', where as per Algorithm \ref{alg:FFBH2} the factors that are divided through are then unused further in the algorithm.
\item When cancellations occur, it is most likely best to do so via ``intelligent GCDs''. In practice this may require some heuristics to ensure you get the correct content as to end up with correct output - it is easy to verify that you obtain the correct cancellation degree in easy cases, such as matrices of uniform degree of a binary size, such as is assumed in this paper. But in general polynomial matrices may not be of uniform degree nor binary size, and this complicates trying to compute such factors intelligently.
\item Whether such an approach is feasible with respect to any cases obtained in practice for polynomial matrix inversion. Certainly the recursion issue seems a damning one, where one may require recursions on matrices approximately as large as the input, but this becomes more prevalent at very large matrix sizes. On smaller examples (which are likely more abundant), this approach may very optimistically beat the classical case where arithmetic would occur on higher degree polynomials than this work would achieve.
\end{itemize}

\printbibliography

\end{document}